\documentclass[11pt,a4paper]{article}   
\usepackage{anysize}
\usepackage[english]{babel}
\usepackage[shortlabels]{enumitem} 
\usepackage{utf8math}
\usepackage{fontaxes}
\usepackage{trimspaces}
\usepackage{nccfoots}
\usepackage{setspace}
\usepackage{inconsolata}
\usepackage{libertine}
\usepackage{amsmath,amssymb,amsthm}
\usepackage{mathtools}
\usepackage{comment}
\usepackage[all,defaultlines=4]{nowidow}
\usepackage{microtype}
\usepackage{mathrsfs}
\usepackage{mathtools}
\usepackage{marvosym}
\usepackage{thm-restate}
\usepackage{tikz}
\usepackage{bm}
\usepackage{xcolor}
\usetikzlibrary{arrows.meta,patterns}
\usetikzlibrary{ipe} 

\newcommand{\intcone}{\operatorname{intcone}}
\newcommand{\rank}{\operatorname{Rank}}
\newcommand{\Mod}[1]{\mathrm{mod-}#1}

\definecolor{darkgreen}{rgb}{10,117,28}
\usepackage{wrapfig}
\usepackage{diagbox}
\usepackage{stackengine}
\usepackage[ruled,vlined]{algorithm2e}
\usepackage{algorithmicx}

\usepackage[colorinlistoftodos]{todonotes}

\DeclareMathOperator{\lcm}{lcm}
\DeclareMathOperator{\conv}{conv}

\definecolor{blue}{rgb}{0.1,0.2,0.5}
\definecolor{brown}{rgb}{0.61,0.6,0.2}
\usepackage[ocgcolorlinks, linkcolor={blue}, citecolor={blue}]{hyperref}

\newtheorem{theorem}{Theorem}
\newtheorem*{theorem*}{Theorem}

\newtheorem{corollary}[theorem]{Corollary}
\newtheorem{proposition}[theorem]{Proposition}

\newtheorem{conjecture}{Conjecture}

\theoremstyle{remark}
\newtheorem{remark}{Remark}

\usepackage{xspace}

\title{A parameterized linear formulation of the integer hull}

\author{
  Friedrich Eisenbrand\thanks{EPFL, Switzerland, \texttt{friedrich.eisenbrand@epfl.ch}. Work carried out while the author was visiting the University of Washington.} \and
   Thomas Rothvoss \thanks{
 University of Washington, USA, 
  \texttt{rothvoss@uw.edu}. Supported by NSF grant 2318620 \emph{AF: SMALL: The Geometry of Integer Programming and Lattices}.}
}

\date{\today}

\begin{document}

\maketitle



\begin{abstract}
  \noindent 
  Let $A ∈ ℤ^{m ×n}$ be an integer matrix with entries  bounded by $Δ$ in absolute value. Cook et al.~(1986) have shown that there exists a universal matrix $B ∈ ℤ^{m' ×n}$ with the following property:  For each  $b ∈ ℤ^m$, there exists a $t ∈ ℤ^{m'}$ such that the integer hull of the polyhedron $P = \{ x ∈ ℝ^n ： Ax ≤b\}$ is described by $P_I = \{ x ∈ ℝ^n ： Bx ≤t\}$. Our \emph{main result} is that $t$ is an \emph{affine} function of $b$ as long as $b$ is from a fixed equivalence class of the lattice $D ⋅ ℤ^m$. Here $D ∈ ℕ$ is a number that depends on $n$ and $Δ$ only. Furthermore,   $D$ as well as the matrix $B$ can be computed in time depending on $n$ and $\Delta$ only. 
  An application of this result is the solution of an open problem posed by Cslovjecsek et al.~(SODA 2024) concerning the complexity of \emph{2-stage-stochastic integer programming} problems.  
  The main tool of our proof
  is  the  classical theory of \emph{Chv\'atal-Gomory cutting planes} and the \emph{elementary closure} of rational polyhedra. 
\end{abstract}

\section{Introduction}
\label{sec:introduction}

An \emph{integer program} is an optimization problem of the form
\begin{equation}
  \label{eq:8}
  \max\{ c^T x ： Ax≤b, \, x ∈ ℤ^n \},
\end{equation}
where $c ∈ ℤ^n$, $A ∈ℤ^{m × n}$ and $b ∈ ℤ^m$. Many optimization problems can be modeled and solved as an integer program, see, e.g.~\cite{NemhauserWolsey88,W98,conforti2014integer,Schrijver86}.
Unlike \emph{linear programming}, which can be solved in polynomial time \cite{Khachiyan79}, integer programming is known to be NP-complete~\cite{BoTr76}. 

An important special case 
arises if the \emph{polyhedron} $P = \{x ∈ ℝ^n ： Ax ≤ b\}$ is \emph{integral}, i.e., if $P$ is equal to its \emph{integer hull} $P_I = \conv(ℤ^n ∩ P)$. Here $\conv(X)$ denotes the \emph{convex hull} of a set $X ⊆ ℝ^n$.   In this case, the integer program~\eqref{eq:8} can be solved in polynomial time with linear programming and integer linear algebra~\cite{Schrijver86}. For example,  if the matrix $A$ is \emph{totally unimodular}, i.e., if each sub-determinant of $A$ is  $0$ or $\pm 1$, then $P$ is integral for each right-hand-side vector $b ∈ ℤ^m$. Here, the notion  \emph{sub-determinant} of $A$ refers to the determinant of some square sub-matrix of $A$. Efficient algorithms for integer programs defined by matrices $A$ with a fixed sub-determinant bound is an active area of  research~\cite{artmann2017strongly,fiorini2022integer,aprile2024integer}.

\medskip

In the following we consider polyhedra defined by the matrix $A ∈ℤ^{m ×n}$  with varying right-hand-side $b ∈ℤ^m$ and denote the thereby \emph{parameterized} polyhedron by $P(b) = \{ x ∈ ℝ^n ： Ax ≤ b\}$. 
Cook, Gerards, Schrijver and Tardos~\cite{MR839604} have shown that there exists a universal integral matrix $M ∈ ℤ^{ m' ×n}$ that depends only on the matrix $A$ with the following property:
\begin{quote}
  For each right-hand-side $b ∈ℤ^m$ there exists a right-hand-side $t ∈ ℤ^{m'}$ with 
  \begin{equation}
    \label{eq:11}    
    P(b)_I = \{x ∈ ℝ^n ： Mx ≤ t\}.
      \end{equation}
\end{quote}
Furthermore, Cook et al.~\cite{MR839604} show that 
$  \|M\|_∞ ≤ n^{2n} s^n$. 
Here $s$ is an upper bound on the sub-determinants of $A$ and $ \|M\|_∞$ denotes the largest absolute value of an entry of $M$. 

\subsection*{Contribution}

We prove that there is a choice\footnote{Our choice of the matrix $M$ is not  the same as in \cite{MR839604} --- in fact the upper bound on $\|M\|_{\infty}$ arising from our proof is triple-exponential while \cite{MR839604} provides a single-exponential bound.} for $M$ so that the right-hand-side $t ∈ℤ^{m'}$ in~\eqref{eq:11} depends \emph{linearly} 
on the right hand side $b ∈ℤ^m$ for all $b$  belonging to the same equivalence class of a certain lattice that depends on $n$ and $\Delta$ only where $\|A\|_{\infty} \leq \Delta$. We make the natural assumption that $A∈ℤ^{m ×n}$ does not have repeating rows. Recall that the number $m$ of rows is bounded by $m≤(2 Δ +1)^n$ in this case.  
Our \emph{main contribution} is the following. 
\begin{theorem}
  \label{thr:main}
  Let $A ∈ℤ^{m ×n}$ with non-repeating rows and  $\|A\|_∞ ≤Δ$. There exists a $D ∈ ℕ$, and matrices $B ∈ℤ^{m'×n}$ and $C ∈ℤ^{m'×m}$  depending on $A$ such that, for each $r ∈ \{0,\dots,D-1\}^m$, there exist an $f_r∈ℤ^{m'}$ 
  such that the following holds:
  \begin{quote}    
  For each $b \in \mathbb{Z}^m$ with $b - r ∈ {D ⋅ ℤ^m}$, one has
  \begin{equation}
    \label{eq:1}
    P(b)_I = \left\{x ∈ ℝ^n ： Bx ≤ f_r + Cb \right\}.
  \end{equation}
\end{quote}
Furthermore, the number $D∈ ℕ$, the matrices $B$ and $C$, as well as the vector $f_r$ can 
be computed in time depending on $n$ and $\Delta$ only\footnote{That means the algorithm does not depend on $b$.}. In fact, $D \leq 2^{n^{(n \Delta)^{O(n^2)}}}$ suffices. 
\end{theorem}

\begin{remark}
  \label{rem:3}  
  The condition that $A ∈ ℤ^{m ×n}$ has non-repeating rows can also be dropped. In this case  
  the running time depends on $Δ$, $n$, and $m$. 
\end{remark}

\subsubsection*{Applications}
\label{sec:applications}

 We describe three almost immediate applications of Theorem~\ref{thr:main}.

\noindent 
i)  A \emph{2-stage stochastic} integer programming problem is an integer program~\eqref{eq:8} of the form  $\max c^Tx + \sum_{i=1}^n d_i^Ty_i$ subject to
  \begin{displaymath}
    U_i x + V_i y_i = b_i, \, i=1,\dots,n  \, \text{ and }\, x,y_1,\dots,y_n ∈ ℤ_{≥0}^k. 
  \end{displaymath}
  If $Δ$ denotes the largest absolute value of an entry of the integral constraint matrices $U_i,V_i ∈ ℤ^{k×k}$, then the problem can be solved in time $f(k,Δ)$-times a polynomial in the input encoding~\cite{hemmecke2003decomposition,CslovjecsekEPVW-ESA21,ComplexityOf2StageStochIPMathProgB}. 
See also~\cite{jansen2023double} for a doubly exponential lower bound on $f$, assuming the Exponential Time Hypothesis (ETH).  Recently Cslovjecsek et al.~\cite{DBLP:conf/soda/CslovjecsekKLPP24}  have proved that the \emph{integer feasibility problem} of 2-stage stochastic integer programming can be solved in time $f(k,Δ)$-times a polynomial in the input encoding. This time,  $Δ$ is a bound on the largest absolute value of the matrices $V_i$ \emph{only}. The authors pose the open problem  whether such a complexity bound follows for the full problem  (optimization version) as well. We answer this question in the affirmative in Section~\ref{sec:optimizing-over-2}.

\medskip

\noindent
ii) Cslovjecsek et al.~\cite{DBLP:conf/soda/CslovjecsekKLPP24} provide an important structural result on the \emph{convexity of integer cones}. Recall that the columns of a matrix $W ∈ ℤ^{m ×n}$ generate a \emph{cone} denoted by
$\textrm{cone}(W) = \{ Wx : x \in \mathbb{R}_{\geq 0}^n \}$, an \emph{integer cone} $\intcone(W) = \{ Wx : x \in \mathbb{Z}_{\geq 0}^n \}$ and a \emph{lattice} $\Lambda(W) = \{ Wx : x \in \mathbb{Z}^n\}$.
Clearly $\intcone(W) \subseteq \textrm{cone}(W) \cap \Lambda(W)$. The reverse however is not true in general\footnote{Equality even fails in dimension $m=1$. For example for $W = (2,3)$ one has $\textrm{cone}(W)=[0,\infty)$, $\Lambda(W) = \mathbb{Z}$ while $1 \notin \intcone(W)$.}.
However, \cite{DBLP:conf/soda/CslovjecsekKLPP24} prove that there is a $D \in \mathbb{N}$ depending on $m$ and $\|W\|_{\infty}$ so that for the integer cone intersected with any shift of the lattice $Λ = D ⋅ ℤ^m$, convexity holds.
More precisely, for each $r ∈ \{0,\dots,D-1\}^m$ there exists a polyhedron $Q_r ⊆ ℝ^m$ such that
\begin{displaymath}
  (r + Λ) ∩ \intcone(W)  = (r+Λ) ∩ Q_r. 
\end{displaymath}
Their original proof is fairly involved. It is implied by our main theorem very quickly as we lay out in Section~\ref{sec:conv-integ-cones}.

\medskip

\noindent
iii) Finally, we turn our attention to \emph{4-block} integer programming problems for which there only exist polynomial-time algorithms, where the degree of the polynomial depends on the parameters $k$ and $Δ$~\cite{DBLP:conf/ipco/HemmeckeKW10}. 
We prove that in time $f(k,\Delta)$ times a polynomial of fixed degree one can at least find an \emph{almost} feasible solution that violates the $k$ many \emph{linking constraints} by an \emph{additive} constant depending on $k$ and $Δ$. The details are in Section~\ref{sec:an-almost-solution}.

\subsubsection*{Proof idea}

{The proof} of Theorem~\ref{thr:main} is based on the classical  theory of \emph{Chv\'atal-Gomory} cutting planes and the  \emph{elementary closure} of polyhedra. We will review this  theory in Section~\ref{sec:gomory-cutt-plan} at a level of detail that is necessary for us here. The reader shall also be referred to the seminal textbooks~\cite{Schrijver86,conforti2014integer}.

\medskip 
However, the main idea can be easily described by resorting to the following main principles of the theory of cutting planes: 
The \emph{elementary closure} of a rational polyhedron $P$ is  again a rational polyhedron $P'$ that satisfies
$
  P_I ⊆ P' ⊆ P 
  $. While given a point $x$, testing if $x \in P'$ is coNP-complete~\cite{1999Eisenbrand99}, the computation of $P'$ is fixed-parameter-tractable in $n$ and $\Delta$. 
Furthermore,   if $P^{(i)}$ denotes the result of $i$~successive  applications of the closure operator, then the \emph{Chv\'atal rank} of $P$ is the smallest $i∈ ℕ$  with $P^{(i)} = P_I$. It is always finite and, for us very important, it is bounded by a function on $n$ and $Δ$~\cite{MR839604}. 
Together, this indicates that, in order  to prove our main result, it is enough to show some analogous statement for  $P^{(i)}$ instead of proving it directly for $P_I$. This is done in Theorem~\ref{thr:2} and constitutes the heart of the proof.


%



\section{Chv\'atal-Gomory cutting planes}
\label{sec:gomory-cutt-plan}


Let $P = \{x ∈ ℝ^n ： Ax ≤b\}$ be a polyhedron where $A ∈ ℝ^{m ×n}$, $b ∈ ℝ^m$. For any $c∈ ℤ^n$ and $δ ≥ \max\{c^Tx ： Ax≤b\}$,  the inequality $c^Tx ≤ ⌊δ⌋$ is called a \emph{Chv\'atal-Gomory cutting plane}~\cite{Gomory58,Chvatal73}, or briefly \emph{cutting plane} of $P$.
It is valid for all integer points $x ∈  P ∩ ℤ^n$ and therefore also for the integer hull $P_I$ of $P$.
\begin{figure}[h] 
  {
    \resizebox{12cm}{!}{%
      \centering
    \input{./Ipe/Cutting-Plane.tex}
  }
}
\caption{The valid inequality $-x_1+x_2 ≤ δ$ yields the cutting plane $-x_1+x_2 ≤ ⌊δ ⌋ $.}
  \label{fig:1}
\end{figure}

\noindent
For the ease of notation we will also write $c^Tx \leq \delta$ for the set of points $\{ x \in \mathbb{R}^n : c^Tx \leq \delta\}$.
The intersection of all cutting planes of $P$ is called the \emph{elementary closure} and it is denoted by
\begin{equation}
  \label{eq:2}  
  P' = \bigcap_{\substack{ (c^Tx ≤δ) ⊇ P \\ c ∈ℤ^n}}\big( c^Tx ≤ ⌊δ⌋ \big). 
\end{equation}
A cutting plane $c^Tx ≤⌊δ⌋$ follows from a simple inference rule that is based on linear programming duality.  Assuming that $P \neq \emptyset$, duality implies that there exists a $λ ∈ ℝ^m_{≥0}$ such that
\begin{enumerate}[i)]
\item $λ^T A = c^T$ and
\item $λ^T b ≤ δ$ hold.   
\end{enumerate}
See e.g. Cor 7.1h in Schrijver~\cite{Schrijver86}.
Therefore, the elementary closure can be described as 
\begin{equation}
  \label{eq:3}  
  P' = \bigcap_{\substack{  λ ∈ ℝ^m_{≥0} \\ λ^T A ∈ℤ^n}} \big( (λ^T A) x ≤ ⌊λ^Tb⌋ \big).
\end{equation}
If the polyhedron $P$ is rational, then the matrix $A$ in the representation $Ax ≤ b$, as well as $b$  can be chosen to be integral, i.e., $A ∈ℤ^{m ×n}$, $b ∈ ℤ^m$. If this is the case, then $P' ⊆P$ and  $P'$ is a rational  polyhedron as well~\cite{Schrijver80}. We repeat the argument below and provide a useful bound on the $ℓ_∞$-norm of the left-hand-side vector $c$ of  a non-redundant cutting plane $c^Tx ≤⌊δ$⌋. 
\begin{theorem}[Schrijver~\cite{Schrijver80}] \label{thm:Schrijver1980}
  Let $P = \{x ∈ ℝ^n ： Ax≤b\}$ be a polyhedron with $A ∈ ℤ^{m ×n}$, $\|A\|_∞≤Δ$ and $b ∈ ℤ^m$.  
  Then 
   \begin{displaymath}
  P' = \bigcap_{\substack{  λ ∈ [0,1]^m \\ λ^T A ∈ℤ^n} } \big( (λ^T A) x ≤ ⌊λ^Tb⌋ \big). 
\end{displaymath}
In particular, in~\eqref{eq:2},  $P'$  is described by all cutting planes $c^Tx ≤ ⌊δ⌋$ with $c ∈ ℤ^n$ and $\|c\|_∞≤ n ⋅ Δ$. 
\end{theorem}

\begin{proof}
  Consider a cutting plane
  $c^T x ≤ ⌊δ⌋$ derived as 
  $(λ^TA )x ≤⌊λ^T b⌋$ with $λ ∈ ℝ_{≥0}^m$, $λ^T A = c^T$ and $λ^T b ≤δ$. 
  The inequality $(λ- ⌊λ⌋)^T A x ≤ ⌊(λ- ⌊λ⌋)^T b ⌋$ is again a cutting plane, since $(λ- ⌊λ⌋)^T A ∈ ℤ^n$. We have
  \begin{displaymath}
    λ^T A x = (λ- ⌊λ⌋)^T A x +  ⌊λ⌋^T A x ≤ ⌊(λ- ⌊λ⌋)^T b ⌋ + ⌊λ⌋^T b = ⌊λ^Tb⌋. 
  \end{displaymath}
  This implies that the inequality $(λ^TA )x ≤⌊λ^T b⌋$ is linearly implied by the system $Ax≤b$ and the inequality $(λ- ⌊λ⌋)^T A x ≤ ⌊(λ- ⌊λ⌋)^T b ⌋$.

  Furthermore, we can assume that $λ$ is an optimal  vertex solution to the linear program
\begin{equation}
  \label{eq:LPforLambda}
    \min\{ λ^T b ： λ^T A = c^T, \, λ ∈ ℝ^m_{≥0} \}, 
  \end{equation}
  which implies that $λ$ has only $n$ nonzero entries. From there, it follows that
  \begin{displaymath}
    \|(λ- ⌊λ⌋)^T A\|_∞ ≤ n ⋅Δ. 
  \end{displaymath}
\end{proof}

\noindent
By applying the closure operator $i$-times successively, one obtains $P^{(i)}$, 
 the \emph{$i$-th elementary closure} of  $P⊆ ℝ^n$. The next corollary is immediate. 
\begin{corollary}
  \label{co:1}
  Let $P = \{x ∈ ℝ^n ： Ax≤b\}$ be a rational polyhedron with $A ∈ ℤ^{m ×n}$, $\|A\|_∞≤Δ$ and $b ∈ ℤ^m$.  Let $i ∈ \mathbb{Z}_{\geq 0}$. 
  Then  $P^{(i)}$ is described by a system of inequalities
  \begin{displaymath}
    P^{(i)} = \{ x ∈ ℝ^n ：Cx ≤ d\},
  \end{displaymath}
  $C∈ ℤ^{m'×n}$, $d ∈ ℤ^{m'}$ with $\|C\|_∞ ≤ n^i Δ$. 
\end{corollary}

\subsubsection*{The rank of rational polyhedra}
\label{sec:rank-rati-polyh}

The \emph{Chv\'atal rank}~\cite{Schrijver80} of  a polyhedron $P$ is the smallest natural number $i ∈ \mathbb{Z}_{\geq 0}$ such that $P^{(i)} = P_I$. Schrijver~\cite{Schrijver80} has shown that the Chv\'atal rank of a rational polyhedron is always finite. 


Cook et al.~\cite{CookCoullardTuran87} showed that the Chv\'atal rank of integer empty polyhedra is bounded by a function on the dimension. More precisely, they show that, the rank of $P ⊆ ℝ^n$ with $P ∩ ℤ^n = ∅$ is bounded by a function $g(n)$ that satisfies the recursion 
\begin{equation}
  \label{eq:12}
  g(n) ≤ ω(n) ( 1 + g(n-1)) +1, 
\end{equation}
where $ω(n)$ is the \emph{flatness constant} in dimension $n$. The best known bound on $ω(n)$ is $O(n \log^3(n))$, see~\cite{reis2023subspace}, which gives a bound of $g(n) \leq n^{O(n)}$. 


Most relevant for us is another result of Cook et al.~\cite{MR839604} that states that the rank can be bounded by a function of $n$ and $\Delta$. It follows from a proximity theorem, the bound on $g(n)$~\eqref{eq:12} and the bound on the $ℓ_∞$-norm of the facet normal-vectors of $P_I$, see also~\cite[Theorem~23.4]{Schrijver86}. More precisely, they show that the rank of $P(b)$ is bounded by
\begin{equation}
  \label{eq:13}  
  \max \left\{ g(n), n^{2n +2} s^{n+1} +1 + n^{2n +2} s^{n+1} g(n-1)\right\} 
\end{equation}
where $g(n)$ is as in~\eqref{eq:12} and $s$ is an upper bound on the largest sub-determinant of $A$.  With the Hadamard bound on $s$ one can see that there exists a constant $c ∈ ℕ$ such that~\eqref{eq:13} is bounded from above by $(n ⋅Δ)^{c ⋅n^2}$. We summarize this in the following theorem.  

\begin{theorem}[Cook, Gerards, Schrijver and Tardos~\cite{MR839604}]
  \label{thr:1}
  Let $P = \{x ∈ ℝ^n ： Ax ≤b\}$ be a rational polyhedron with $A ∈ ℤ^{m ×n}$, $\|A\|_∞ ≤ Δ$ and $b ∈ ℤ^m$. There exists a function $\rank(n,Δ)$ such that
  \begin{displaymath}
    P^{(i)} = P_I
  \end{displaymath} 
  for each $i≥ \rank(n,Δ)$. Moreover, there exists a constant $c ∈ ℕ$ such that 
  \begin{equation}
    \label{eq:14}
    \rank(n,Δ) ≤ (n ⋅ Δ)^{c ⋅ n^2}. 
  \end{equation}
\end{theorem}

\section{Proof of the main theorem}
\label{sec:proof-main-theorem}

This section contains the proof of Theorem~\ref{thr:main}.
Throughout, we assume that  $A ∈ℤ^{m ×n}$ satisfies $\|A\|_∞ ≤Δ$. 
The strategy is  to show an analogous statement to Theorem~\ref{thr:main} for the $i$-th elementary closure.

\subsection{Warmup --- Linearity of the first elementary closure}
\label{sec:line-elem-clos}


To warm up, we prove the linearity statement from Theorem~\ref{thr:main} for $P(b)'$ instead of $P(b)_I$.
While Theorem~\ref{thr:mainForFirstClosure} will not be needed to prove the main theorem, its statement and proof contain the main ideas that  can be described with simpler notation. 
\begin{theorem}
  \label{thr:mainForFirstClosure}
  Let $A ∈ℤ^{m ×n}$ with non-repeating rows and  $\|A\|_∞ ≤Δ$. There exists a $D ∈ ℕ$, $B ∈ℤ^{m'×n}$ and $C ∈ℤ^{m'×m}$ such that, for each $r ∈ \{0,\dots,D-1\}^m$, there exist an $f_r∈ℤ^{m'}$ 
  such that the following holds:
  \begin{quote}    
  For each $b \in \mathbb{Z}^m$ with $b - r ∈ {D ⋅ ℤ^m}$, one has
  \begin{equation}
    \label{eq:1}
    P(b)' = \left\{x ∈ ℝ^n ： Bx ≤ f_r + Cb \right\}.
  \end{equation}
\end{quote}
\end{theorem}
\begin{proof}
Consider a cutting plane $c^Tx ≤ ⌊δ⌋$, where $c ∈ℤ^n$. If this cutting plane is non-redundant, then there exists an optimal vertex solution $λ∈ ℝ^m_{≥0}$  of the linear program \eqref{eq:LPforLambda} with $δ ≥ λ^Tb$. 
This means that there are at most $n$ linearly independent rows of $A$, indexed by $J⊆\{1,\dots,m\}$ such that $λ_J$ is the unique solution of the linear equation
\begin{displaymath}
  λ_J^T A_J = c^T,
\end{displaymath}
with all other entries of $λ$ equal to zero. By deleting linearly dependent columns of $A_J$ and the corresponding entries of $c$, this shows that there is a $|J|×|J|$ \emph{non-singular} sub-matrix $\overline{A}$ of $A$ with
\begin{displaymath}
  λ_J^T \overline{A} ∈ℤ^{|J|}. 
\end{displaymath}
Cramer's rule implies then that one can write
$\lambda = \mu/D$ with $μ ∈ \mathbb{Z}_{\geq 0}^m$, where $D$ is any integer multiple of $\det(\bar{A})$.
However, in order to have the same denominator for all vectors $c$, we choose $D$ as
 an integer multiple of all sub-determinants of $A$.

 \smallskip 
 \noindent 
 Furthermore, we have seen in Theorem~\ref{thm:Schrijver1980} that we can replace $λ∈ ℝ_{≥0}^m$ by $λ - ⌊λ⌋$
 or by $e_j$ for some $j \in [m]$. Thus we have
 \begin{displaymath}
λ = μ/D \,\text{ with }\,   \mu \in \{0,\ldots,D\}^m. 
\end{displaymath}
A Chv\'atal-Gomory cut $(μ/D)^T A x ≤ ⌊(μ/D)^Tb⌋$ with $\mu \in \{0,\ldots,D\}^m$ and $μ^TA ≡ 0 \pmod{D}$ is called $\Mod{D}$ cut~\cite{CapraraFischettiLetchford00}. 
 In other words, if $D ∈ ℕ$ is an integer multiple of the sub-determinants of $A$, then  the set of $\Mod{D}$ cuts contains all non-redundant cutting planes. 

\smallskip
Let us now \emph{fix}  a $μ \in \{0,\ldots,D\}^m$ with $μ^TA≡ 0 \pmod{D}$ as well as a \emph{remainder} $r \in \{ 0,\ldots,D-1\}^m$ and consider  $b \in \mathbb{Z}^m$
with $b \equiv r \pmod{D}$.  The cutting plane derived by $λ = μ/D$ can be written as
\begin{equation}\label{eq:16}
  \begin{array}{rcl}    
  \underbrace{\left({\mu}/{D}\right)^TA}_{\text{row of }B}x & \leq &  \lfloor (\mu/D)^Tb \rfloor \\
                                                              & = &   \underbrace{\lfloor (\mu^Tr)/D \rfloor - (\mu^Tr)/D}_{\text{entry of } f_r}   + \underbrace{(\mu^T/D)}_{\text{row of }  C} b,
                                                                      \end{array}
                                                                    \end{equation}
where the equation above follows from the fact that $μ^T(b-r)/D$ is an integer.                                                 
We conclude that the right-hand-side of the cut induced by $λ = μ/D$ is 
indeed linear in $b$.

Note that the above constructed $C$ and $f_r$ are rational and not-necessarily integral. In order to be conform with the claim, we thus  scale with $D$ to obtain integral data. Then using row indices $T := \{ \mu \in \{0,\ldots,D\}^m \mid μ^TA ≡ 0 \pmod{D} \}$ we may choose
\[
 B := (\mu^TA)_{\mu \in T}, \quad f_r := (D\lfloor (\mu^Tr)/D \rfloor - (\mu^Tr))_{\mu \in T}, \quad \textrm{and} \quad C := (\mu^T)_{\mu \in T}
\]
for $r \in \{ 0,\ldots,D-1\}^m$.
\smallskip
\end{proof}

\subsection{Linearity of the $i$-th elementary closure}
\label{sec:linearity-i-th}
We now proceed with proving the statement from Theorem~\ref{thr:main} for $P(b)^{(i)}$ for any $i$. 
Recall that the Chv\'atal rank of $P(b)$ is bounded by $\rank(n,Δ)$ 
and that  each $P(b)^{(i)}$ can be described by a constraint matrix with infinity norm bounded by $n^{\min\{i,\rank(n,Δ)\}} Δ$. We  define $D ∈ ℕ$ to be the least common multiple of nonzero sub-determinants of $n×n$-integer matrices with entries bounded by $n^{\rank(n,Δ)} ⋅ Δ$ in absolute value. In other words, let
\begin{equation}
  \label{eq:15}  
    D := \lcm \left\{ \det(C) ： C ∈ℤ^{k ×k},\, k≤n, \,  \|C \|_∞ ≤ n^{\rank(n,Δ)} ⋅Δ, \, \det(C) ≠ 0\right\}. 
  \end{equation}
  It follows that $P^{(i+1)}$ is defined by $\Mod{D}$ cuts derived from the description of $P^{(i)}$ as in Corollary~\ref{co:1}. In other words, $D$ is the largest modulus that is ever necessary in the derivation of a valid cutting plane. 

\begin{theorem}
  \label{thr:2}
  Let $A ∈ℤ^{m ×n}$ with non-repeating rows and $\|A\|_∞ ≤Δ$ and let $D$ be as in \eqref{eq:15}. For each $i ∈ \{ 0,\ldots,\rank(n,Δ)\}$ and $r ∈  \mathbb{Z}^m$, there exist $B ∈ ℤ^{m'×n}$, $C∈ℤ^{m'×m}$ and $f_r ∈ ℤ^{m'}$ such that for each $b ∈ℤ^m$ with $b -r ∈ D^i ⋅ ℤ^m$, one has 
  \begin{equation}
    \label{eq:1}
    P(b)^{(i)} = \left\{x ∈ ℝ^n ： B\, x ≤f_r + \frac{ C (b-r)}{D^i} \right\}.
  \end{equation}
  Furthermore, $B$ satisfies $\|B\|_∞ ≤ n^i Δ$. 
\end{theorem}

\begin{proof}
 
 We prove the assertion by induction on $i$. For $i=0$, one has
  \begin{displaymath}
     P(b)^{(0)} = P(b) = \{  x ∈ ℝ^n ： Ax ≤b \}.
   \end{displaymath}
   Each $b ∈ ℤ^m$ is congruent to $\bm{0}$ modulo ${1=D^0}$. We set, for $r = \bm{0}$, 
   \begin{displaymath}
     B := A, \quad f_r := \bm{0}, \quad \text{and} \quad C := I.  
   \end{displaymath}

\medskip 
   \noindent   
   Next, we assume that the assertion is true for $i \in \{ 0,\ldots,\rank(n,Δ)-1\}$ 
   and show that it is true for $i+1$ as well. 

   \medskip 
   \noindent   
  Let $r ∈ \mathbb{Z}^m$.  By the induction hypothesis, one has 
   \begin{equation} \label{eq:MainProofIH}
     P(b)^{(i)} = \left\{ x ∈ ℝ^n ： Bx ≤ f_r + \frac{C (b-r)}{D^i} \right\} 
   \end{equation}
for all $b \in \mathbb{Z}^m$ with $b ≡ r \pmod{D^{i}}$   where $B ∈ ℤ^{m' ×n}$ with $\|B\|_{\infty} \leq  n^i Δ$.
Now fix a $b \in \mathbb{Z}^m$ with $b ≡ r \pmod{D^{i+1}}$. In particular also $b ≡ r \pmod{D^{i}}$ and so
 \eqref{eq:MainProofIH} still holds for $b$.

We now construct the cutting planes of $P(b)^{(i)}$. 
By the choice of $D$ in \eqref{eq:15} and the argument from the proof of Theorem~\ref{thr:mainForFirstClosure} we
know that a non-redundant cutting plane of $P(b)^{(i)}$ is of the form
   \begin{equation}
     \label{eq:6}
     (μ/D)^T B x ≤ \left⌊ (μ/D)^T \left(f_r + \frac{C (b-r)}{D^i}\right)\right⌋,
   \end{equation}
   where  $μ ∈ \{0,\dots,D\}^{m'}$ is  of support  at most $n$ and $μ^T B ≡ 0 \pmod{D}$.
We wish to describe integer  matrices $B'$, $C'$ and an integer vector $f_r'$ such that 
  \begin{displaymath}
     P(b)^{(i+1)} = \left\{ x ∈ ℝ^n ： B'x ≤ f_r' + \frac{ C' (b-r)}{D^{i+1}}\right\} 
   \end{displaymath} 
  Since $b ≡ r \pmod{D^{i+1}}$, the right-hand-side of~\eqref{eq:6} is 
   \begin{eqnarray*}
     \left⌊ (μ/D)^T \left(f_r + \frac{C(b-r)}{D^i} \right)\right⌋ & = & \left⌊ (μ/D)^T f_r   +  \frac{ μ^T C (b-r)}{D^{i+1}} \right⌋ \\
     & = & \underbrace{\left⌊ (μ/D)^T f_r \right⌋}_{\text{ entry of } f'_r}   +  \frac{ \overbrace{μ^T C}^{\text{row of } C'_r} (b-r)}{D^{i+1}},
   \end{eqnarray*}
   where the last equation follows from the fact that each entry of $b-r$ is divisible by $D^{i+1}$ and therefore that
   \begin{displaymath}
     \frac{ μ^T C (b-r)}{D^{i+1}} ∈ℤ. 
   \end{displaymath}
   Hence, using row indices $T := \{ \mu \in \{ 0,\ldots,D\}^{m'} \mid \mu^TB \equiv 0 \pmod{D} \}$
   we can choose
   \[
  B' := ((\mu/D)^TB)_{\mu \in T}, \quad f_r' :=  (\left⌊ (μ/D)^T f_r \right⌋)_{\mu \in T}, \quad \textrm{and} \quad C' := (\mu^TC)_{\mu \in T}
   \]
   By construction, $B',C'$ and $f_r'$ are integer and $\|B'\|_{\infty} \leq n\|B\|_{\infty} \leq n^{i+1}\Delta$.
   
   
\end{proof}

We note that the condition in Theorem~\ref{thr:2} that $A ∈ℤ^{m × n}$ has non-repeating rows can be dropped. In this case, we have a dependence on $n$, $m$ and $Δ$.  


\begin{proof}[Proof of Theorem~\ref{thr:main}]
  We abbreviate $k := \rank(n,\Delta)$. Let $r \in \mathbb{Z}^m$.
 By the definition of Chv\'atal rank and Theorem~\ref{thr:2}, for all $b \in \mathbb{Z}^m$ with $b \equiv r \pmod{D^k}$ one has 
  \begin{eqnarray*}
    P(b)_I = P(b)^{(k)} &=& \Big\{ x \in \mathbb{R}^n : Bx \leq f_r + \frac{C(b-r)}{D^{k}} \Big\} \\
    &=& \big\{ x \in \mathbb{R}^n : D^kBx \leq (D^kf_r-Cr) + Cb \big\} 
  \end{eqnarray*}
  where $B$ and $C$ are some integer matrices, $f_r$ is an integer vector $f_r$ and $D$ is defined in \eqref{eq:15}.
  We note that the proof of Theorem~\ref{thr:2} provides a finite and constructive procedure to compute $D$, the matrices $B$, $C$ and the vector $f_r$.
  Then the choices to satisfy the claim of Theorem~\ref{thr:main} are
  \[
   B' := D^kB, \quad f_r' := D^kf_r-Cr, \quad C' := C, \quad \textrm{and} \quad D' := D^k.
 \]
  Using the Hadamard bound\footnote{The Hadamard bound says that for any $M \in \mathbb{R}^{n \times n}$ one has $|\det(M)| \leq (\sqrt{n}\|M\|_{\infty})^n$.} and $k = \rank(n,Δ)≤ (n⋅Δ)^{C ⋅n^2}$ from
  Theorem~\ref{thr:1}, the number $D'$ can be bounded from above by 
  \begin{eqnarray*}
    D'=D^k 
    & ≤ &  \lcm\left\{ 1,\dots,  n^{{(n ⋅Δ)}^{c'⋅n^2}} \right\}^{{(n ⋅Δ)}^{c⋅n^2}} \\
      & ≤ & 2^{n^{{(n ⋅Δ)}^{c'' ⋅ n^2}}},
\end{eqnarray*}
where $c', c'' ∈ ℕ_+$ are suitable constants. This bound is triple-exponential in $n$. 
\end{proof}

\section{Applications}
\label{sec:applications-1}

\subsection{The convexity of integer cones}
\label{sec:conv-integ-cones}

Recall that for a matrix $W \in \mathbb{Z}^{m \times n}$, we denote the \emph{integer cone} as $\intcone(W) = \{ Wx : x \in \mathbb{Z}_{\geq 0}^n \}$. An integer cone is a discrete set and in general it will have ``holes'', i.e.,  integer points in the rational cone of $W$ that are not in the integer cone. But Cslovjecsek, Kouteck{\'{y}}, Lassota, Pilipczuk and Polak~\cite{DBLP:conf/soda/CslovjecsekKLPP24} proved that integer cones are indeed convex in a discrete sense when intersected with shifts of certain sparse lattices.
\begin{theorem}[Convexity of Integer Cones~\cite{DBLP:conf/soda/CslovjecsekKLPP24}] \label{thm:ConvexityOfIntegerCones}
  Let $W \in \mathbb{Z}^{m \times n}$ with $\|W\|_{\infty} \leq \Delta$. Then there is a $D \in \mathbb{N}$ dependent only on $m$ and $\Delta$ so that for every $r \in \{ 0,\ldots,D-1\}^m$ there exists a polyhedron $Q_r \subseteq \mathbb{R}^m$ so that
  \[
  \Lambda_r \cap \intcone(W) = \Lambda_r \cap Q_r
\]
where $\Lambda_r = r + D ⋅\mathbb{Z}^m$.
\end{theorem}
The original proof of this result takes a substantial amount of work. We demonstrate that it is quickly implied by our main result.
\begin{proof}[Proof of Theorem~\ref{thm:ConvexityOfIntegerCones}]
  We can assume that $W ∈ℤ^{m ×n}$ does not have repeated columns. Therefore, $n$ is bounded by $(2 ⋅  Δ +1)^m$. 
  Consider $P(b) = \{ x \in \mathbb{R}^n : Wx = b, x \geq \bm{0}\}$. Observe that the condition $Wx = b$ and $ x \geq \bm{0}$ can be formulated in the usual inequality standard-form as the conjunction of $Wx ≤ b$, $-Wx ≤ -b$ and $-Ix ≤ \bm{0}$. 

  Let $D$ be as in Theorem~\ref{thr:main} depending on $Δ$ and $n$ (and thus $m$) only.  Now, fix any $r \in \{ 0,\ldots,D-1\}^m$ and let $B,C,f_r$  be the vectors so that
  \[
    P(b)_I = \{ x \in \mathbb{R}^n : Bx \leq f_r + Cb \}
  \]
  for all $b \in \mathbb{Z}^m$ with $b \equiv r \pmod{D}$. Let $Q_r = \{ b \in \mathbb{R}^m \mid \exists x \in \mathbb{R}^n: Bx \leq f_r + Cb\}$. Then $Q_r$ is the linear projection of a polyhedron and hence it is again a polyhedron. 

  Now, consider a right hand side $b∈ ℤ^m$ with $b \equiv r \pmod{D}$. Then
  \begin{eqnarray*}
    b \in \intcone(W) &\Leftrightarrow& P(b) \cap \mathbb{Z}^n \neq \emptyset \\
                               &\Leftrightarrow& P(b)_I \neq \emptyset \\
    &\Leftrightarrow& b \in Q_r.
  \end{eqnarray*}
\end{proof}
We should point out however that a closer look into \cite{DBLP:conf/soda/CslovjecsekKLPP24} reveals a double-exponential bound for $D$, while our construction gives a triple-exponential bound on $D$.

\subsection{Optimizing over 2-stage stochastic IPs}
\label{sec:optimizing-over-2}

A \emph{2-stage stochastic integer program} is of the form
\begin{eqnarray*}
  \max c^Tx + \sum_{i=1}^n d_i^Ty_i & & (\textsc{2SSIP}) \\
                                        U_ix + V_iy_i &=& b_i \quad \forall i=1,\ldots,n \\
                                        x,y_1,\ldots,y_n &\in& \mathbb{Z}_{\geq 0}^k, 
\end{eqnarray*}
with integer matrices $U_i,V_i ∈ℤ^{k ×k}$. 
It will be useful to think of $k$ as a constant and $n$ as large. Then after determining the $k$
variables in $x$, the rest of the IP decomposes into $n$ disjoint parts, again with $k$ variables each. 
It was proven in Cslovjecsek et al~\cite{DBLP:conf/soda/CslovjecsekKLPP24} that the \emph{decision} version
(without an objective function) can be solved in time $g(k,\max_i \|V_i\|_{\infty})$ times a fixed polynomial in the encoding length. The reader may note that this running time has only a polylogarithmic dependence on $\|U_i\|_{\infty}$.
The argument of \cite{DBLP:conf/soda/CslovjecsekKLPP24} is based on Theorem~\ref{thm:ConvexityOfIntegerCones}
and the authors left it as an open problem whether the same FPT-type
running time is possible for the optimization variant as stated above in $({\textsc{2SSIP}})$.
Intuitively, given a matrix $A$ and remainder $r$, Theorem~\ref{thm:ConvexityOfIntegerCones} provides a polyhedral description of the right hand sides $b$ with $b \equiv r\pmod{D}$ for which there exists a $x \in \mathbb{Z}_{\geq 0}^n$ with $Ax = b$. In contrast, our main result gives a polyhedral
description of the \emph{pairs} $(x,b)$ with $Ax = b$ where $x \in \mathbb{Z}_{\geq 0}^n$ and $b \equiv r \pmod{D}$.
This gives us access to the coefficient vector $x$ which is needed in the objective function.

We solve their open problem in the affirmative:
\begin{theorem} \label{thm:2SSIPalgorithm}
The problem $(\textsc{2SSIP})$ can be solved in time $g(k,\Delta)$ times a fixed polynomial in the encoding length where $\Delta = \max\{ \|V_i\|_{\infty} : i=1,\ldots,n\}$.
\end{theorem}
\begin{proof}
We use the following algorithm:
\begin{enumerate}
\item[(1)] Let $P_i(b') = \{ y_i \in \mathbb{R}^k \mid V_iy_i = b', y_i \geq \bm{0} \}$ for $i=1,\ldots,n$ and $b' \in \mathbb{Z}^k$.
\item[(2)] Let $D := D(k,\Delta)$ be the parameter as in Theorem~\ref{thr:main}. 
\item[(3)] Guess $r \in \{ 0,\ldots,D-1\}^k$ so that $x^{**} \equiv r \pmod D$ where $(x^{**},y_1^{**},\ldots,y_n^{**})$ is an optimum integral solution\footnote{The algorithm only needs to know $r$ and not the solution $(x^{**},y_1^{**},\ldots,y_n^{**})$. The existence of the solution is enough.} to $(\textsc{2SSIP})$.
\item[(4)] For all $i=1,\ldots,n$, let $r_i \in \{ 0,\ldots,D-1\}^k$ with $r_i \equiv b_i - U_ir \pmod{D}$.
\item[(5)] For each $i \in [n]$, apply Theorem~\ref{thr:main} to write $P_i(b')_I = \{ y_i \in \mathbb{R}^k \mid B_iy_i \leq f_i+C_ib'\}$ for all $b' \in \mathbb{Z}^k$ with $b' \equiv r_i \pmod{D}$.
\item[(6)] Compute an optimum extreme point solution $(x^*,y_1^*,\ldots,y_n^*)$ to the mixed integer linear program   \begin{eqnarray*}
    \max c^Tx + \sum_{i=1}^n d_i^Ty_i & & \\
    B_iy_i &\leq& f_i + C_i(b_i - U_ix) \quad \forall i=1,\ldots,n\\
     x &\equiv& r \pmod{D} \\
     x &\in& \mathbb{Z}_{\geq 0}^k
  \end{eqnarray*}
\item[(7)] Return $(x^*,y_1^*,\ldots,y_n^*)$
\end{enumerate}
We know that $b_i - U_ix^{**} \equiv r_i \pmod{D}$ and so $(x^{**},y_1^{**},\ldots,y_n^{**})$ is feasible for the MIP in (6).
The MIP in $(6)$ has only $k$ many integral variables (and $nk$ many fractional ones) and hence it
can be solved within the claimed running time.
On the other hand, if $(x^*,y_1^*,\ldots,y_n^*)$ is an optimum extreme point to (6), then $x^* \in \mathbb{Z}_{\geq 0}^k$ and each $y^{*}_i$ is an extreme point to $P_i(b_i-U_ix^*)_I$. Each such extreme point must be integral. That shows the claim. 
\end{proof}
%

\subsection{An almost solution to the 4-block problem}
\label{sec:an-almost-solution}
The \emph{4-block} problem is the IP of the form
\begin{eqnarray*}
  \max c^Tx + \sum_{i=1}^n d_i^Ty_i & & ({\textsc{4BlockIP}}) \\
  Wx + X_1y_1 + \ldots + X_ny_n &=& a \quad (*) \\
                                        U_ix + V_iy_i &=& b_i \quad \forall i=1,\ldots,n \\
                                        x,y_1,\ldots,y_n &\in& \mathbb{Z}_{\geq 0}^k 
\end{eqnarray*}
where $U_i,V_i,W,X_i \in \mathbb{Z}^{k \times k}$ and $a,b_i \in \mathbb{Z}^k$.
This is a strict generalization of $({\textsc{2SSIP}})$. Note that after deleting $k$ many variables \emph{and} $k$ many constraints, the problem decomposes into disjoint IPs with $k$ variables each. Indeed,
$(\textsc{4BlockIP})$ can be solved in time $n^{f(k,\Delta)}$ times a fixed polynomial in the encoding length where $\Delta = \max\{ \|U_i\|_{\infty},\|V_i\|_{\infty},\|W\|_{\infty},\|X_i\|_{\infty} : i \in [n] \}$, see the work of Hemmecke, K\"oppe and Weismantel~\cite{DBLP:conf/ipco/HemmeckeKW10}. It is a popular open
problem in the theoretical IP community whether there is an FPT-type algorithm as well. 
\begin{conjecture}[4-block conjecture]
  $(\textsc{4BlockIP})$ can be solved in time $f(k,\Delta)$ times a fixed polynomial in the encoding length
  where $\Delta$ is the largest absolute value appearing in $U_i,V_i,W,X_i$. 
\end{conjecture}
We prove that one can at least find an \emph{almost} feasible solution that violates the $k$ many joint constraints by an additive constant (assuming $k$ and the maximum entries are bounded by constants).
We require the following fact:
\begin{theorem}[Proximity bound\label{thm:ProximityBound}] 
  Let $P = \{ x \in \mathbb{R}^n \mid Ax \leq b\}$ where $A \in \mathbb{Z}^{m \times n}$, $b \in \mathbb{R}^m$ and $\|A\|_{\infty} \leq \Delta$ for $\Delta \in \mathbb{N}$.
  Then for any $y \in P_I$ and $c \in \mathbb{R}^n$ there is a $z \in P \cap \mathbb{Z}^n$ with $c^Tz \geq c^Ty$
  and $\|z-y\|_{\infty} \leq (\Delta n)^{O(n^3)}$.
\end{theorem}
This bound is a variation of similar statements proven in Cook, Gerards, Schrijver and Tardos~\cite{MR839604}.
For the sake of completness, we give a full proof in the Appendix.
Now we prove the following: 
\begin{theorem}
  Suppose that $(\textsc{4BlockIP})$ is feasible with value $OPT$ and let $\Delta = \max\{ \|V_i\|_{\infty} : i \in [n]\}$. Then in time $f(k,\Delta)$ times a polynomial in the input length one can find a vector $(x^*,y_1^*,\ldots,y_n^*)$ with $c^Tx^* + \sum_{i=1}^n d_i^Ty_i^* \geq OPT$
  that satisfies all constraints except $(*)$. Moreover, the error for constraint $(*)$ is bounded by
  \[
  \big\|(Wx^* + X_1y_1^* + \ldots + X_ny_n^*) - a\big\|_{\infty} \leq g(k,\Delta,\max_{i=1,\ldots,n} \|X_i\|_{\infty})
  \]
\end{theorem}
\begin{proof}
  We use a similar argument to Theorem~\ref{thm:2SSIPalgorithm}. Again set $P_i(b') := \{ y_i \in \mathbb{R}^k \mid V_iy_i = b', y_i \geq \bm{0} \}$ for $i=1,\ldots,n$ and $b' \in \mathbb{Z}^k$. Let $(x^{**},y_1^{**},\ldots,y_n^{**})$ be a solution to $({\textsc{4BlockIP}})$ attaining the value of $OPT$. Again we can guess a vector $r$ so that $x^{**} \equiv r \pmod D$
  where $D$ is a parameter so that by Theorem~\ref{thr:main} one has
  \[
    P_i(b')_I = \{ y_i \in \mathbb{R}^k \mid B_iy_i \leq f_i+C_ib'\}
  \]
  for all $b' \in \mathbb{Z}^k$ with $b' \equiv r_i \pmod{D}$ where $r_i \equiv b_i - U_ir \pmod{D}$.
  Consider the mixed integer linear program
\begin{eqnarray*}
  \max c^Tx + \sum_{i=1}^n d_i^Ty_i & & ({\textsc{4BlockMIP}}) \\
  Wx + X_1y_1 + \ldots + X_ny_n &=& a \quad (**) \\
  y_i &\in& P_i(b_i-U_ix)_I  \quad \forall i=1,\ldots,n \\
  x &\equiv& r \pmod{D} \\
                                        x &\in& \mathbb{Z}_{\geq 0}^k 
\end{eqnarray*}
We compute an optimum extreme point solution $(x^*,\tilde{y}_1,\ldots,\tilde{y}_n)$ to $({\textsc{4BlockMIP}})$.
We note that $x^*$ will be integral, but as all vertices of $P_i(\cdot)_I$ are integral and $(**)$
are only $k$ non-trivial constraints, the set $J =\{ i \in [n] \mid \tilde{y}_i \notin \mathbb{Z}^k \}$ has size
$|J| \leq k$. Then by the proximity bound from Theorem~\ref{thm:ProximityBound}, for each $i \in J$, there is a $y_i^*  \in P_i(b_i-U_ix^*)_I \cap \mathbb{Z}^k$ with $\|y_i^*-\tilde{y}_i\|_{\infty} \leq (\Delta k)^{O(k^3)}$ and $d_i^Ty_i^* \geq d_i^T \tilde{y}_i$. We set $y_i^* = \tilde{y}_i$ for all $i \notin J$. Then by the triangle inequality
\[
 \big\|(Wx^* + X_1y_1^* + \ldots + X_ny_n^*) - a\big\|_{\infty} \leq \sum_{i \in J} \|X_i(\tilde{y}_i-y_i^*)\|_{\infty} \leq k \cdot k \cdot (\Delta k)^{O(k^3)} \cdot \max_{i=1,\ldots,n} \|X_i\|_{\infty}.
\]
Hence $(x^*,y_1^*,\ldots,y_n^*)$ satisfies the claim.
\end{proof}


\paragraph{Acknowledgements.} The authors would like to thank the anonymous SODA reviewers for their
valuable input.

\bibliography{papers,mybib,books,my_publications,extended}

\appendix

\section{A proximity bound for the integer hull}

A \emph{proximity bound} describes the fact that fractional solutions in a polyhedron cannot be
arbitrarily far away from integer solutions as long as some integer solution exists.
In this section, we fill in the proof for Theorem~\ref{thm:ProximityBound} that we used earlier.
For a vector $x \in \mathbb{R}^n$ we abbreviate $\lfloor x \rfloor := (\lfloor x_1 \rfloor,\ldots,\lfloor x_n \rfloor)$.
First, we need an auxiliary proximity bound that only applies to integer polytopes. We would like to point
out that the line of arguments used here is almost identical to the one in the proof of \cite[Theorem~1]{MR839604}.
\begin{proposition} \label{prop:ProximityForIntegerPolytope}
  Let $A \in \mathbb{Z}^{m \times n}$, $b \in \mathbb{R}^m$, $c \in \mathbb{R}^n$ and let $P := \{ x \in \mathbb{R}^{n} \mid Ax \leq b\}$ be a bounded polyhedron with 
   $P = P_I$. Then for any $x^* \in P$  there is an $x^{**} \in P \cap \mathbb{Z}^n$ with $c^Tx^{**} \geq c^Tx^*$ and $\|x^{**}-x^*\|_{\infty} \leq n \cdot s$ where $s$ is the largest subdeterminant of $A$.
\end{proposition}
\begin{proof}
  Since $P$ is integral and bounded, there is an integral optimum solution $z^*$ to the LP $\max\{ c^Tx \mid x \in P\}$.
  By optimality, $c^Tz^* \geq c^Tx^*$.
  
We split the matrix $A$ into $A'$ and $A''$ so that
  \[
   A'(z^*-x^*) \geq \bm{0} \quad \textrm{and} \quad A''(z^*-x^*) < \bm{0}
 \]
 We also split  $b = (b',b'')$ accordingly.
 Consider the cone
 \[
   C := \big\{ x \in \mathbb{R}^n \mid A'x \geq \bm{0}\textrm{ and }A''x \leq \bm{0}\big\}
 \]
 Let $y$ be an optimum dual solution, i.e. $y^TA = c^T$, $y^Tb = c^Tz^*$ and $y \geq \bm{0}$ (see e.g. \cite[Cor~7.1g]{Schrijver86}).
 By complementary slackness (see e.g.~\cite[Section~7.9]{Schrijver86}), the solution $y$ cannot use any row where $z^*$ has slack, i.e. $y'' = \bm{0}$
 and hence $(y')^TA' = c^T$.
 For any $x \in C$ we can verify that
 \[
   c^Tx = \underbrace{(y')^T}_{\geq \bm{0}}\underbrace{A'x}_{\geq \bm{0}} \geq 0
 \]
 In other words, $C$ contains only directions in which the objective function improves (though the cone may not contain all such directions).
 
 Let $B \in \mathbb{R}^{n \times k}$ be the matrix whose columns are the generators of the cone, i.e.  $C = \textrm{cone}(B)$. Using Cramer's rule, the matrix $B$ can be chosen to be
 integral with $\|B\|_{\infty} \leq s$. We note that by construction  $z^*-x^* \in C$.  By Carath\'eodory's Theorem (\cite[Cor~7.1i]{Schrijver86}), there is a $y \in \mathbb{R}_{\geq 0}^k$ with $z^*-x^* = By$ and $|\textrm{supp}(y)| \leq n$. We want to prove that
 \[
  x^{**} := x^* + B (y-\lfloor y \rfloor ) = z^* - B\lfloor y \rfloor 
 \]
 satisfies the claim. Clearly $x^{**} \in \mathbb{Z}^n$ because $z^*$, $B$ and $\lfloor y \rfloor$ are integral.
 As for the distance to $x^*$ we have that
 \[
 \|x^{**}-x^*\|_{\infty} = \|B(y-\lfloor y \rfloor )\|_{\infty} \leq \sum_{j \in \textrm{supp}(y)} \underbrace{\|B^j\|_{\infty}}_{\leq s} \cdot \underbrace{(y_j-\lfloor y_j \rfloor )}_{\leq 1} \leq n ⋅ s
 \]
 as $|\textrm{supp}(y)| \leq n$, where $B^j$ denotes the $j$-th column of $B$. It remains to show that $x^{**} \in P$. For the first group of constraints
 we have
 \[
   A'x^{**} = \underbrace{A'z^*}_{\leq b'} - \sum_{j=1}^k \underbrace{\lfloor y_j\rfloor}_{\geq 0} \cdot \underbrace{A'B^j}_{\geq \bm{0}} \leq b'
 \]
 as the columns of $B$ are in the cone $C$.
 For the second group of constraints we have
 \[
 A''x^{**} = \underbrace{A''x^*}_{\leq b''} + \sum_{j=1}^k \underbrace{(y_j-\lfloor y_j\rfloor)}_{\geq 0} \cdot \underbrace{A''B^j}_{\leq \bm{0}} \leq b''
 \]
 Hence $Ax^{**} \leq b$ and the claim is proven.
\end{proof}

Now we can restate and prove Theorem~\ref{thm:ProximityBound}.
\begin{theorem*}[Proximity bound --- Theorem~\ref{thm:ProximityBound}] 
  Let $P = \{ x \in \mathbb{R}^n \mid Ax \leq b\}$ where $A \in \mathbb{Z}^{m \times n}$, $b \in \mathbb{R}^m$ and $\|A\|_{\infty} \leq \Delta$ for $\Delta \in \mathbb{N}$.
  Then for any $y \in P_I$ and $c \in \mathbb{R}^n$ there is a $z \in P \cap \mathbb{Z}^n$ with $c^Tz \geq c^Ty$
  and $\|z-y\|_{\infty} \leq (\Delta n)^{O(n^3)}$.
\end{theorem*}
\begin{proof}
  We may assume that $P$ is bounded --- if not then intersect $P$ with any large enough bounding box that incudes
  the at most $n+1$ integer points that have $y$ in its convex hull.  
  Let $s$ be the largest subdeterminant of $A$. Recall that by the Hadamard bound one has  $s \leq (\Delta n)^n$.
  We can write $P_I = \{ x \in \mathbb{R}^n \mid A'x \leq b'\}$ with $A'$ integer and $\|A'\|_{\infty} \leq n^{2n}s^n \leq (\Delta n)^{O(n^2)}$, see \cite[Theorem~7]{MR839604}.
  The largest subdeterminant of $A'$ is $s' \leq (\|A'\|_{\infty}n)^n \leq (\Delta n)^{O(n^3)}$.
  As $y \in P_I$, Prop~\ref{prop:ProximityForIntegerPolytope} provides a $z \in P_I \cap \mathbb{Z}^n$ with $c^Tz \geq c^Ty$ and $\|z-y\|_{\infty} \leq n \cdot s' \leq (\Delta n)^{O(n^3)}$. This concludes the claim.
\end{proof}

\end{document}